\documentclass{cccg25}
\usepackage{graphicx,amssymb,amsmath}

\usepackage{import}
\usepackage{xcolor}
\usepackage{cleveref}

\newtheorem{definition}[theorem]{Definition}

\usepackage{lineno}

\crefname{conj}{conjecture}{conjectures}
\crefname{obs}{observation}{observations}
\crefname{prop}{proposition}{propositions}
\crefname{lemma}{lemma}{lemmas}
\crefname{definition}{definition}{definitions}

\DeclareMathOperator\sm{sm}
\AddToHook{env/lemma/begin}{\crefalias{theorem}{lemma}}
\AddToHook{env/prop/begin}{\crefalias{theorem}{prop}}
\AddToHook{env/definition/begin}{\crefalias{theorem}{definition}}

\title{Straight-line Orthogonal Drawing of Complete Ternary Tree Requires $O(n^{1.032})$ Area}

\author{Hong Duc Bui\thanks{\texttt{buihd@u.nus.edu}}}

\begin{document}
\thispagestyle{empty}
\maketitle

\begin{abstract}
	We resolve a conjecture posed by Covella, Frati and Patrignani in \cite{Covella_2021}
	by proving the straight-line orthogonal drawing of the complete ternary tree with $n$ nodes
	satisfying the subtree separation property
	with smallest area has area $\Omega (n^{1.031})$.
	We also improve the upper bound of this area to $O(n^{1.032})$.
\end{abstract}

\section{Introduction}

We consider the problem of embedding a tree into a grid.
Given a tree $T$ and a grid $G$,
an embedding maps each vertex $v$ of $T$ to a distinct vertex $v'$ of $G$,
and each edge $u v$ of $T$ to a polyline $u' v'$ of $G$,
such that no two polylines intersect except at endpoints.

We call such an embedding \emph{orthogonal}, because all grid lines are either horizontal or vertical.
If, furthermore, the embedding satisfies that every edge $u v$ gets mapped to either all horizontal segments
or all vertical segments, then we call such an embedding \emph{straight-line}.
In the literature, straight-line orthogonal embeddings in a grid are also called straight-line orthogonal \emph{drawings}.

Determining whether a straight-line orthogonal drawing of a tree exists is simple ---
each node of the grid has degree at most $4$, so a necessary condition is that
each node of the tree has degree at most $4$, and it can be shown that this condition is also sufficient.
As such, most research on drawings of trees is concerned with minimizing
the \emph{area} of the grid, where the area is defined as the number of nodes of the grid.

There has been much research on this problem, see \cref{table_summary_existing_works} for a summary.
If the drawing is not required to be straight-line (only orthogonal),
\cite{Valiant_1981} proves that there exists an embedding with area $O(n)$ for all embeddable trees ---
this is seen to be asymptotically optimal, because the area must be at least $n$.
When the drawing is required to be straight-line, the complete binary tree can be drawn in
$O(n)$ area by \cite{Crescenzi_1992}, and it is proven in \cite{Chan_2019} that any binary tree
can be drawn in $n \cdot 2^{O(\log^* n)}$ area, which is almost linear.

\begin{table}
	\centering
\begin{tabular}{ccccc}
	                & \!\!\!Straight-line& Upper bound           & Ref. \\ \hline
	Comp.\ Binary & $\checkmark $          & $O(n)$                & \cite{Crescenzi_1992} \\
	Binary          & $\checkmark $          & $n \cdot 2^{O(\log^* n)}$ & \cite{Chan_2019} \\
	Comp.\ Ternary& $\checkmark $          & $O(n^{1.118})$        & \cite{Ali2015} \\
	Ternary         & $\checkmark $          & $O(n^{1.576})$        & \cite{Covella_2021} \\
	Any             &              & $O(n)$                & \cite{Valiant_1981}
\end{tabular}
\caption{Summary of existing works on tree drawings.}
\label{table_summary_existing_works}
\end{table}

For the ternary case however, the known bounds are less tight.
Prior to our work, the
best known upper bound for the complete ternary tree is $O(n^{1.118})$ proven in \cite{Ali2015},
improving upon an existing bound $O(n^{1.262})$ proven in \cite{Frati}.
The best known upper bound for an arbitrary ternary tree is $O(n^{1.576})$ proven in \cite{Covella_2021}.

In this article, we improve the upper bound of the area needed from $O(n^{1.118})$ to $O(n^{1.032})$.
Our method is based on the analysis in \cite{Covella_2021} of the drawings
satisfying the subtree separation property.
Drawings with this property are more easily analyzed.

We also improve the lower bound of the area needed in the special case of drawings
satisfying the subtree separation property to $\Omega (n^{1.031})$.
This is the first non-trivial lower bound on the area needed,
with the trivial lower bound being $\Omega (n)$.

This article is organized as follows.
In \cref{sec_definitions}, we formally define the notations being used in this article.
In \cref{sec_motivation}, we show the result of a numerical experiment
that motivates the proof.
In \cref{sec_upper_bound}, we explain the general proof strategy,
and prove a weaker upper bound $O(n^{1.051})$ for demonstration.
In \cref{sec_lower_bound}, we use a very similar proof strategy
to prove the lower bound $\Omega (n^{1.031})$.
Finally, in \cref{sec_numerical_upper}, we describe our numerical algorithm
to provide a certificate of the upper bound $O(n^{1.032})$.

\section{Definitions}
\label{sec_definitions}

We define the notation for the complete ternary tree following \cite{Covella_2021}.

\begin{definition}
	For each positive integer $l$, let $T_l$ be the rooted complete ternary tree with $l$ layers ---
	that is, each non-leaf node has exactly $3$ children, and each root-to-leaf path
	has exactly $l$ nodes.
\end{definition}
With this definition, $T_1$ has $1$ node, $T_2$ has $4$ nodes, etc.

We have defined straight-line orthogonal drawings of a tree in the introduction.
Now we will formally define the subtree separation property.

\begin{definition}
	A drawing is said to satisfy the subtree separation property if, for every nodes $a$ and $b$ of the tree
	such that the two subtrees rooted at $a$ and $b$ have no nodes in common,
	the smallest axis-aligned bounding rectangles in the drawing
	containing all the nodes of these two subtrees
	have no grid nodes in common.
\end{definition}

See \cref{fig_subtree_not_separated} for an illustration of a drawing that does not satisfy the subtree separation property.

\begin{figure}
    \centering
    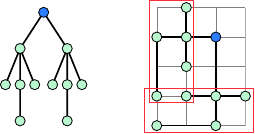
    \caption{Example of a straight-line orthogonal drawing that does not satisfy the subtree separation property.
	The tree (left panel, root marked in blue) is embedded in a $4 \times 5$ grid (right panel), and the two subtrees rooted at the two children of the root have intersecting bounding rectangles.}
    \label{fig_subtree_not_separated}
\end{figure}

We define the following notation for convenience.

\begin{definition}
	Given positive integers $l$ and $h$, odd positive integer $w$, write $T_l \leq (w, h)$
	if there is an orthogonal straight-line drawing of $T_l$ in a grid with width $w$ and height $h$
	such that:
	first, the subtree separation property is satisfied;
	second, the root of the tree is on the middle vertical grid line;
	and third, the vertical ray from the root to the top of the grid does not
	intersect any tree nodes or edges.
\end{definition}

See \cref{fig_def_embed_operator} for an illustration that $T_3 \leq (5, 6)$.
The red segment in the figure marks the vertical ray from the root to the top of the grid.
In order for the drawing to satisfy the third condition of the definition above,
no nodes or edges can intersect this red segment.

\begin{figure}
    \centering
    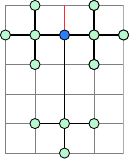
    \caption{Illustration for $T_3 \leq (5, 6)$.}
    \label{fig_def_embed_operator}
\end{figure}

We define a special class of constructions as follows, which has the advantage of being very easy to analyze.
This is a slightly modified form of a 1-2 drawing in \cite[Section 3]{Covella_2021}.

\begin{definition}
	\label{def_symmetric_construction_12}
	We call a straight-line orthogonal drawing of $T_l$ a \emph{symmetric 1-2 drawing} if the following conditions are satisfied.
	For $l=1$, the only symmetric 1-2 drawing is the unique drawing on the $1 \times 1$ grid.
	For $l>1$, let $\Gamma _l$ and $\Gamma _b$ be two symmetric 1-2 drawings of $T_{l-1}$,
	then:\footnote{The $l$ in $\Gamma _l$ is a literal character, not an index.}
	\begin{itemize}
		\item define a drawing $\Gamma _1$ created by \emph{construction 1} as follows:
			put a copy of $\Gamma _b$ below the root at distance $1$,
			put two copies of $\Gamma _l$ rotated $90^{\circ }$ to the left and right of the root
			at the minimum distance such that the subtree separation property is satisfied.
		\item define a drawing $\Gamma _2$ created by \emph{construction 2} as follows:
			put two copies of $\Gamma _l$ rotated $90^{\circ }$ to the left and right of the root
			at distance $1$ from the root,
			then put a copy of $\Gamma _b$ below the root
			at the minimum distance such that the subtree separation property is satisfied.
	\end{itemize}
\end{definition}

It is easier to understand this looking at a picture than looking at the description,
see \cref{fig_definition_construction_12}.

\begin{figure}
    \centering
    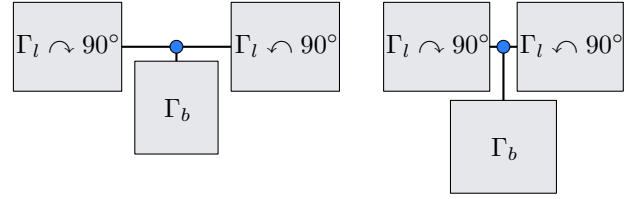
	\caption{Illustration of \cref{def_symmetric_construction_12}, with the left panel illustrating construction 1 and the right panel illustrating construction 2.}
    \label{fig_definition_construction_12}
\end{figure}

From the definition, we get the following lemma, which also explains the name.
\begin{lemma}
	\label{claim_recursion_formula}
	All symmetric 1-2 drawings have odd width, and are vertically symmetric.
	Furthermore, let the size of $\Gamma _l$ be $(w_l, h_l)$ and the size of $\Gamma _b$ be $(w_b, h_b)$,
	then the size of $\Gamma _1$ is $(2 h_l+w_b, \frac{w_l}{2}+\max(\frac{w_l}{2}, h_b+\frac{1}{2}))$,
	and the size of $\Gamma _2$ is $(\max(2 h_l+1, w_b), w_l + h_b)$.
\end{lemma}

As we have mentioned, the symmetric 1-2 drawings are very easy to analyze.
In particular, we can compute \emph{all} grid sizes $(w, h)$ such that $T_l \leq (w, h)$.
The algorithm to compute these grid sizes was given in \cite[Lemma 5]{Covella_2021}:
\begin{lemma}
	For a fixed $l$, the Pareto-optimal pairs $(w, h)$ at level $l$
	can be computed in time polynomial in the number of nodes of $T_l$.
\end{lemma}
We should explain what Pareto-optimal pairs mean in this lemma.
Because if $w \leq w'$ and $h \leq h'$ then $T_l \leq (w, h) \Longrightarrow T_l \leq (w', h')$,
it suffices to consider for each $l$ the pairs $(w, h)$ such that $T_l \leq (w, h)$
and there exist no pair $(w', h')$ such that $w' \leq w$, $h' \leq h$, $w' \cdot h' < w \cdot h$, and $T_l \leq (w', h')$.
We call these pairs \emph{Pareto-optimal} at level $l$.

Apart from being easy to analyze, the symmetric 1-2 drawings additionally satisfy the following properties,
which is proven in \cite[Lemma 3]{Covella_2021}.
\begin{lemma}
	Given any straight-line orthogonal drawing $\Gamma $ of the complete ternary tree $T_l$,
	there exists a symmetric 1-2 drawing $\Gamma '$ whose both width and height are no more than
	those of $\Gamma $.
\end{lemma}

\section{Motivation: The Pattern of the Pareto-optimal Grid Sizes}
\label{sec_motivation}

We compute the Pareto-optimal grid sizes for small values of $l$:
\begin{itemize}
	\item When $l = 1$, the only pair is $(1, 1)$.
	\item When $l = 2$, the only pair is $(3, 2)$.
	\item When $l = 3$, there is a pair $(5, 5)$ corresponding to construction 2, and a pair
		$(7, 4)$ corresponding to construction 1.
\end{itemize}
We make a scatterplot for all the pairs for each value of $l$.
The result is shown in \cref{fig_pareto_optimal_scatterplot},
where both $x$-axis and $y$-axis use a logarithmic scale.

\begin{figure}
    \centering
	\includegraphics{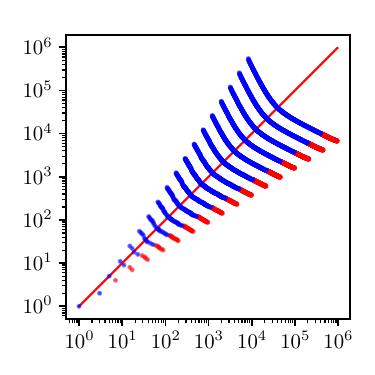}
    \caption{A scatterplot of all Pareto-optimal grid sizes. Red points denote construction 1, blue points denote construction 2.}
    \label{fig_pareto_optimal_scatterplot}
\end{figure}

From the figure, the pattern is obvious. Our goal is thus to prove that the pattern continues indefinitely.

In order to do so, we need to look at how these grid sizes were computed ---
the set of Pareto-optimal grid sizes at level $l$ is computed only from the Pareto-optimal grid sizes
at level $l-1$, independent of what happens at earlier levels.
As such, our proof will be inductive --- assume the Pareto-optimal grid sizes
at level $l-1$ satisfy some bound, we prove the Pareto-optimal grid sizes
at level $l$ satisfy another bound.

In order to formalize these concepts, we make the following definitions.
\begin{definition}[$\leq $-relation for grid sizes]
	Let $w$, $h$, $w'$, $h'$ be real numbers.
	We say $(w, h) \leq (w', h')$ if $w \leq w'$ and $h \leq h'$.
	Similarly, $(w, h) \geq (w', h')$ if $w \geq w'$ and $h \geq h'$.
\end{definition}
\begin{definition}
	For any set $A \subseteq \mathbb {R}^2$, define the \emph{upper-closure} $C(A) \subseteq \mathbb {R}^2$
	to be $C(A) = \{(w, h) \in \mathbb {R}^2 \mid \exists (w', h') \in A, (w', h') \leq (w, h) \}$.
	We say a set $A$ is \emph{upper-closed} if $A \subseteq \mathbb {R}^2$ and $C(A) = A$.
\end{definition}
\begin{definition}
	For each $l \geq 1$, define the set $E_l$ to be all pairs $(w, h) \in \mathbb {R}^2$
	such that $T_l \leq (w, h)$.
	Define $S_l = C(E_l)$.
\end{definition}
So for example, at $l = 2$, $E_l$ consists of all pairs of integers $(w, h)$
such that $w$ is odd, $w \geq 3$ and $h \geq 2$, while $S_l$ consists of all pairs of reals $(w, h)$
such that $w \geq 3$ and $h \geq 2$.
We see that $S_l$ is the natural extension of $E_l$ to the domain of all reals.

\begin{definition}
	For any upper-closed set $A \subseteq \mathbb {R}^2$ and real number $\delta $, define the \emph{shift} of $A$ by $\delta $ to be
	$\Delta (A, \delta ) = \{ (w \cdot \exp \delta , h \cdot \exp \delta ) \mid (w, h) \in A \}$.
\end{definition}

\begin{definition}
	For any upper-closed set $A \subseteq \mathbb {R}^2$, define the \emph{advance} of $A$ to be
	\begin{multline*}
		N(A) = C\Bigl(\bigl\{ 
				{\textstyle(2 h_l+w_b, \frac{w_l}{2}+\max(\frac{w_l}{2}, h_b+\frac{1}{2}))}
				\\ \quad \bigm| (w_l, h_l) \in A, (w_b, h_b) \in A
		\bigr\} \\
		\cup \bigl\{
			(\max(2 h_l+1, w_b), w_l + h_b)
			\\ \bigm| (w_l, h_l) \in A, (w_b, h_b) \in A
		\bigr\}\Bigr).
	\end{multline*}
\end{definition}
The formulas can be seen to be directly taken from \cref{claim_recursion_formula}.
As such, we get the following:
\begin{lemma}
	\label{lemma_N_advance_S}
	For each $l \geq 1$, $N(S_l) = S_{l+1}$.
\end{lemma}

\begin{definition}
	For two upper-closed sets $A$ and $B$, we say $A \leq B$ if $A \supseteq B$, and $A \geq B$ if $A \subseteq B$.
\end{definition}
Unfortunately this $\leq $ notation looks reversed (just like in algebraic geometry), but readers should be able to see the motivation
behind this definition ---
for $(w, h)$ and $(w', h') \in \mathbb {R}^2$,
$(w, h) \leq (w', h')$ if and only if $C(\{(w, h)\}) \leq C(\{(w', h')\})$.

\section{Upper Bound: Preliminary}
\label{sec_upper_bound}

Using these definitions, we explain how the proof of the upper bound proceeds.
As previously explained, it will use induction.
\begin{itemize}
	\item Base case: $S_{l_0} \leq T$.
	\item Induction step: If $S_l \leq \Delta (T, d)$ for any $d \geq 0$, then $S_{l+1} \leq \Delta (T, d+\delta )$.
\end{itemize}
By induction, we get $S_l \leq \Delta (T, \delta \cdot (l-l_0))$ for all $l \geq l_0$.
For a suitable set $T$, this implies the complete ternary tree $T_l$
can be embedded in a grid with area $O(e^{2 \delta l}) = O(n^{2 \delta /\log 3})$
where $n$ is the number of nodes in the tree $T_l$.

For the purpose of demonstration, we will use $T = S_{18}$. As such, the base case is trivially satisfied for $l_0 = 18$.
The constant $\delta $ used is $\log(63761/35808)$,
which is in fact the smallest $\delta $ such that $S_{19} \leq \Delta (S_{18}, \delta )$.
This value was computed by using the algorithm described in \cite[Lemma 5]{Covella_2021}
to compute both $S_{18}$ and $S_{19}$ explicitly.
With this value of $\delta $, $2\delta /\log 3 < 1.051$.

Now we prove that the induction step holds.
\begin{lemma}
	\label{lemma_advance_upper_bound}
	Let $d\geq 0$. If $N(T) \leq U$, then $N(\Delta (T, d)) \allowbreak \leq \Delta (U, d)$.
\end{lemma}
\begin{proof}
	Expanding out the definitions, it suffices to prove the following. For both
	\[ f(w_l, h_l, w_b, h_b) = {\textstyle(2 h_l+w_b, \frac{w_l}{2}+\max(\frac{w_l}{2}, h_b+\frac{1}{2}))} \]
	and
	\[ f(w_l, h_l, w_b, h_b) = (\max(2 h_l+1, w_b), w_l + h_b), \]
	for all $d\geq 0$, then
	for all positive real $w_l$, $h_l$, $w_b$, $h_b$,
	set \[ (w, h) = f(w_l, h_l, w_b, h_b), \]
	then
	\[ f(w_l \cdot e^d, h_l \cdot e^d, w_b \cdot e^d, h_b \cdot e^d) \leq (w \cdot e^d, h \cdot e^d). \]
	So for example, in the second case of $f$, for the first component, we need to prove
	\[ 
		\max(2 h_l \cdot e^d+1, w_b \cdot e^d) \leq \max(2 h_l+1, w_b) \cdot e^d.
	\]
	Since $d\geq 0$, $e^d\geq 1$ and $(2 h_l+1) \cdot e^d \geq 2 h_l \cdot e^d +1$.
	Other cases are omitted because they are similar.
\end{proof}
Using this, the induction step can be proven.
Since $N(T) \leq \Delta (T, \delta )$, we get
\[ N(\Delta (T, d)) \leq \Delta (\Delta (T, \delta ), d) = \Delta (T, \delta +d). \]
The induction hypothesis gives us
$S_l \leq \Delta (T, d)$, so \[ N(S_l) = S_{l+1} \leq N(\Delta (T, d)). \]
Combining the two inequalities, we get $S_{l+1} \leq \Delta (T, \delta +d)$ as desired.

To complete the proof, we just need the following.
\begin{lemma}
	Let $T \subseteq \mathbb {R}^2$ be any non-empty upper-closed set, and $l_0$ be a fixed integer.
	If $S_l \leq \Delta (T, \delta \cdot (l-l_0))$ for all positive integers $l \geq l_0$,
	then the area of the smallest grid that $T_l$ can be embedded in is $O(n^{2\delta /\log 3})$,
	where $n \in \Theta (3^l)$ is the number of nodes in $T_l$.
\end{lemma}
\begin{proof}
	Pick an arbitrary fixed element $(w, h) \in T$.
	By definition, $S_l \leq \Delta (T, \delta \cdot (l-l_0))$, so $(w \cdot e^{\delta \cdot (l-l_0)}, h \cdot e^{\delta \cdot (l-l_0)}) \in S_l$.
	As such, there is a grid with area no more than $w \cdot e^{\delta \cdot (l-l_0)} \cdot h \cdot e^{\delta \cdot (l-l_0)}$
	that $T_l$ can be embedded in, this value is $O(n^{2\delta /\log 3})$ as needed.
\end{proof}

\section{Lower Bound}
\label{sec_lower_bound}

Similarly, we will fix a set $S$, a positive real $\delta $, and a positive integer $l_0$, and prove:
\begin{itemize}
	\item Base case: $S_{l_0} \geq S$;
	\item Induction step: For any $l$ and $d$, if $S_l \geq \Delta (S, d)$,
		then $S_{l+1} \geq \Delta (S, d+\delta )$.
\end{itemize}
As such, for every $l \geq l_0$, we get $S_l \geq \Delta (S, (l-l_0) \cdot \delta )$.
For a suitable initial set $S$, this implies the smallest area of a grid that $T_l$ can be embedded in is
$\Omega (\exp(2\delta \cdot l))$.

This time however, the analog of \cref{lemma_advance_upper_bound} would be the following
(we don't need to use this in the article, as such it is not proved):
\begin{quote}
	Let $d \leq 0$. If $N(S) \geq U$, then $N(\Delta (S, d)) \geq \Delta (U, d)$.
\end{quote}
Note that $d \leq 0$. This means inequalities can only be shifted ``backward'', not forward.
As such, we would need to conceptually define a set $S$ ``at infinity'',
then shift it backwards.
To formalize it, we give the following definition.
\begin{definition}
	\label{def_advance_at_inf}
	For any upper-closed set $A \subseteq \mathbb {R}^2$, define the \emph{advance at infinity} of $A$ to be
	\begin{multline*}
		N^\infty (A) = C\Bigl(\bigl\{ 
				{\textstyle(2 h_l+w_b, \frac{w_l}{2}+\max(\frac{w_l}{2}, h_b))}
				\\ \quad \bigm| (w_l, h_l) \in A, (w_b, h_b) \in A
		\bigr\} \\
		\cup \bigl\{
			(\max(2 h_l, w_b), w_l + h_b)
			\\ \bigm| (w_l, h_l) \in A, (w_b, h_b) \in A
		\bigr\}\Bigr).
	\end{multline*}
\end{definition}
This should be thought of as $\lim_{d \rightarrow +\infty } \allowbreak \Delta (N(\Delta (A, d)),-d)$ ---
shift $A$ to ``infinity'', advance it, then shift it back.

Also note that $N^\infty $ is invariant under $\Delta $-shifting ---
formally, for any real $d$, $N^\infty (\Delta (S, d)) = \Delta (N^\infty (S), d)$.

We get the following:
\begin{lemma}
	\label{lemma_advance_lower_bound}
	$N(S) \geq N^\infty (S)$.
	Therefore, if $N^\infty (S) \geq U$, then $N(S) \geq U$.
\end{lemma}

Using reasoning similar to the previous section, the induction proceeds as follows.
Assume the set $S$ satisfies $N^\infty (S) \geq \Delta (S, \delta )$.
Then assume the base case $S_1 \geq S$ holds, the induction step can be proven as follows.
Using \cref{lemma_advance_lower_bound},
\begin{multline*}
N(\Delta (S, d)) \geq N^\infty (\Delta (S, d)) \\ = \Delta (N^\infty (S), d) \geq \Delta (S, d+\delta ).
\end{multline*}
From the induction hypothesis, $S_l \geq \Delta (S, d)$, so
\[ 
S_{l+1} = N(S_l) \geq N(\Delta (S, d)).
\]
Combining the two inequalities, we get $S_{l+1} \geq \Delta (S, d+\delta )$ as needed.

Now, the only remaining challenge is to construct such a set $S$.

\begin{definition}
	Fix constants $\sigma >1$ and $\varepsilon \in \mathbb {R}$.
	Define $S \subseteq \mathbb {R}^2$ to be the upper-closure of the set of points
	$\{ \exp (\omega , \max(\frac{-\omega }{\sigma }, -\omega \cdot \sigma )+\varepsilon ) \mid \omega \in \mathbb {R} \}$.
	Here, we write $\exp (\omega , \eta )$ to denote the pair $(\exp \omega ,\exp \eta )$.
\end{definition}

\begin{lemma}
	\label{lemma_lower_bound}
	For constants $\delta =0.5667$, $\varepsilon =0.10995$ and $\sigma =2.01979$,
	we have $N^\infty (S) \geq \Delta (S, \delta )$.
\end{lemma}
The proof will be deferred for later.
Using the lemma, we get the following:
\begin{theorem}
	Set $\delta $ as above.
	If $T_l \leq (w, h)$, then $w \cdot h \in \Omega (n^{2\delta /\log 3}) \geq \Omega (n^{1.031})$,
	where $n \in \Theta (3^l)$ is the number of nodes in $T_l$.
\end{theorem}
Note that by our definition of $\leq $, the theorem only lower bounds the area of embeddings
satisfying the subtree separation property.

\begin{proof}
	Note that for $l_0 = 2$ then $S_2 \geq S$.

	Apply induction by the plan described above,
	we get $S_l \geq \Delta (S, (l-l_0)\cdot \delta )$ for all $l \geq 2$.

	Note that for every $(w, h) \in S$ then $w \cdot h \geq \exp \varepsilon $,
	therefore for every $(w, h) \in \Delta (S, (l-l_0)\cdot \delta )$
	then $w \cdot h \geq \exp (2 (l-l_0) \delta + \varepsilon )$,
	so we are done.
\end{proof}

Now we prove \cref{lemma_lower_bound}.
\begin{proof}
	Pick $(w, h) \in N^\infty (S)$.
	Define $w_l$, $h_l$, $w_b$, $h_b$ as in \cref{def_advance_at_inf},
	then $(w_l, h_l) \in S$ and $(w_b, h_b) \in S$.

	Expanding out these conditions, we get that the assumptions are, for both $\tau \in \{ \sigma , \frac{1}{\sigma }\}$:
	\begin{align*}
		\log h_l &\geq \varepsilon -\log w_l \cdot \tau , \\
		\log h_b &\geq \varepsilon -\log w_b \cdot \tau .
	\end{align*}
	We need to prove $(w, h) \in \Delta (S, \delta )$.
	This is equivalent to the following statement: for both $\tau \in \{ \sigma , \frac{1}{\sigma }\}$:
	\begin{align*}
		\log h &\geq \delta +\varepsilon +(\delta -\log w) \cdot \tau .
	\end{align*}

	In the first case (construction 1), it suffices for us to prove for both $\tau \in \{ \sigma , \frac{1}{\sigma }\}$:
	\begin{align*}
		\log {\textstyle(\frac{w_l}{2}+h_b)} &\geq \delta +\varepsilon +(\delta -\log (2 h_l+w_b)) \cdot \tau .
	\end{align*}
	For any positive reals $a$ and $b$, $\log(a+b) \geq \max(\log a, \log b)$.
	Define the softmax function $\sm(a, b) = \log(\exp a+\exp b)$, then the left-hand side is $> \log(w_l+h_b) = \sm(\log w_l, \log h_b)$.

	Define $\omega _i = \log w_i$, $\eta _i = \log h_i$ for $i \in \{ l, b \}$.
	Then we need to prove
	\[
		\sm(\omega _l-\log 2, \eta _b) \geq \delta +\varepsilon +(\delta -\sm(\eta _l+\log 2, \omega _b)) \cdot \tau .
	\]
	
	With this new notation, for all $\tau \in \{ \sigma , \frac{1}{\sigma }\}$ and $i \in \{ l, b \}$ then
	\[ \eta _i \geq \varepsilon -\omega _i \cdot \tau .  \]
	Equivalently,
	\[ \omega _i \geq (\varepsilon - \eta _i) \cdot \tau .  \]

	Thus, we just need to prove
	\begin{multline*}
		\sm((\varepsilon -\eta _l)\cdot \tau -\log 2, \varepsilon -\omega _b\cdot \tau ) \\ \geq \delta +\varepsilon +(\delta -\sm(\eta _l+\log 2, \omega _b))\tau .
	\end{multline*}
	Set $\varphi = \eta _l-\omega _b$, 
	since the $\sm$ function satisfies $\sm(a+d, b+d) = \sm(a, b)+d$ for all $a$, $b$, $d \in \mathbb {R}$,
	this simplifies to
	\begin{multline*}
		\sm((\varepsilon -\varphi )\tau -\log 2, \varepsilon ) \\ \geq (1+\tau )\delta +\varepsilon -\sm(\varphi +\log 2, 0)\tau .
	\end{multline*}
	Equivalently, since $1+\tau >0$,
	\[
		\delta \leq \frac{\sm((\varepsilon -\varphi )\tau -\log 2, \varepsilon )-\varepsilon +\sm(\varphi +\log 2, 0)\tau }{1+\tau }.
	\]
	The right hand side only contains one variable $\varphi $.
	We define the function $f_{0, \tau } (\varphi )$ to be the right hand side,
	and we wish to compute the minimum of $f_{0, \tau }$.
	Notice that $f'_{0, \tau } (\varphi ) = 0$ has an unique solution
	\[ \varphi = \frac{(\tau -1)\varepsilon - \log 4}{1 + \tau } \]
	and this can be shown to be the global minimum of $f_{0, \tau }$.
	At this point, the value of $f_{0, \tau }$ is $\geq \delta $ for both $\tau = \sigma $ and $\tau = \frac{1}{\sigma }$,
	so we're done.

	In the second case (construction 2), we need to prove:
	\[
		\log (w_l+h_b) \geq \delta +\varepsilon +(\delta -\log(\max(w_b, 2 h_l))) \cdot \tau .
	\]
	This simplifies to
	\[
		\sm(\omega _l, \eta _b) \geq \delta \cdot (1+\tau )+\varepsilon -\max(\omega _b, \eta _l+\log 2) \cdot \tau .
	\]
	Doing exactly as above, we just need to prove
	\begin{multline*}
		\sm((\tau -1) \varepsilon -\eta _l \tau , -\omega _b \tau ) \\ \geq \delta \cdot (1+\tau )-\max(\omega _b, \eta _l+\log 2) \cdot \tau .
	\end{multline*}
	Set $\varphi = \eta _l-\omega _b$ again, this further simplifies to
	\[
		\delta \leq \frac{\sm(0, (\tau -1) \varepsilon -\varphi \tau )+\max(0, \varphi +\log 2) \cdot \tau }{1+\tau }.
	\]
	Set $f_{1, \tau }(\varphi ) = \frac{\sm(0, (\tau -1) \varepsilon -\varphi \tau )}{1+\tau }$
	and $f_{2, \tau }(\varphi ) = \frac{\sm(0, (\tau -1) \varepsilon -\varphi \tau )+(\varphi +\log 2) \cdot \tau }{1+\tau }$,
	notice that $f_{1, \tau }$ is decreasing, $f_{2, \tau }$ is increasing,
	and the equation $f_{1, \tau }(\varphi ) = f_{2, \tau }(\varphi )$ has a unique solution $\varphi = -\log 2$,
	at this point
	\[ f_{1, \tau }(\varphi ) = f_{2, \tau }(\varphi ) = \frac{\sm(0, (\tau -1) \varepsilon +\tau \log 2)}{1+\tau }. \]
	As such, $\max(f_{1, \tau }(\varphi ), f_{2, \tau }(\varphi ))$ has a global minimum at this point.
	For both $\tau = \sigma $ and $\tau = \frac{1}{\sigma }$, this value is $\geq \delta $.
\end{proof}

We are unable to obtain a closed-form formula for these constants; however, it can be efficiently
computed to arbitrary precision using software such as Mathematica ---
namely, find the value of $\sigma $, $\varepsilon $ and $\delta $
such that
\begin{multline}
	\label{equation_for_exact_constants}
\delta = f_{1, \sigma }(-\log 2)
= f_{1, 1/\sigma }(-\log 2)
\\ = f_{0, 1/\sigma }\Bigl(\frac{(1/\sigma -1) \varepsilon -\log 4}{1+1/\sigma }\Bigr),
\end{multline}
then it can be confirmed that
$f_{0, \sigma }\bigl(\frac{(\sigma -1) \varepsilon -\log 4}{1+\sigma }\bigr) > \delta $.

\section{Numerical Proof for Improved Upper Bound}
\label{sec_numerical_upper}

In \cref{sec_upper_bound}, we improved the upper bound on the minimum area required
for a straight-line orthogonal drawing of $T_l$. In order to do so,
we used a certain upper-closed set $T$ satisfying $N(T) \leq \Delta (T, \delta )$
for a constant $\delta < 0.577$. Specifically, we used $T = S_{18}$.

If we use a different set $T$ such that $N(T) \leq \Delta (T, \delta ')$ for a smaller constant $\delta '$,
we would be able to improve the upper bound accordingly.
In order to simplify the analysis, we use the following.
\begin{lemma}
	Fix an upper-closed set $T$ such that $C(\{ 1, 1 \}) \leq T$, and a constant $\delta >0$.
	If $N^\infty (T) \leq \Delta (T, \delta )$, then for all $\varepsilon >0$,
	we get $N(T') \leq \Delta (T', \delta +\varepsilon )$
	where $T' = \Delta (T, d)$ for some sufficiently large real constant $d$.
\end{lemma}
\begin{proof}
	Unrolling the definition, we need to prove that
	with notation as above,
	for every $(w, h) \in \Delta (T', \delta +\varepsilon )$,
	then $(w, h) \in N(T')$.

	The statement $(w, h) \in \Delta (T', \delta +\varepsilon )$
	is equivalent to $w = \exp(d+\delta +\varepsilon ) w'$, $h = \exp(d+\delta +\varepsilon ) h'$
	for $(w', h') \in T$.

	By assumption, $N^\infty (T) \leq \Delta (T, \delta )$,
	so $N^\infty (\Delta (T, d+\varepsilon )) \leq \Delta (T, d+\delta +\varepsilon )$,
	which means
	$(w, h) \in N^\infty (\Delta (T, d+\varepsilon ))$.
	Expanding out the definition of $N^\infty $,
	this means there exists $(w_l, h_l)$ and $(w_b, h_b) \in \Delta (T, d+\varepsilon )$
	such that either 
	\begin{equation}
		\label{equation_construction_1}
		w = 2h_l+w_b \text{ and } h = \frac{w_l}{2}+\max\Bigl(\frac{w_l}{2}, h_b\Bigr)
	\end{equation}
	or
	\begin{equation}
		\label{equation_construction_2}
		w=\max(2h_l, w_b) \text{ and } h=w_l+h_b.
	\end{equation}

	The statement we need to prove is $(w, h) \in N(T')$.
	With $(w_l, h_l)$ and $(w_b, h_b)$ as above,
	we get
	\begin{equation}
		\label{equation_given_info}
		(w_l/\exp \varepsilon , h_l/\exp \varepsilon ), (w_b/\exp \varepsilon , h_b/\exp \varepsilon ) \in T'.
	\end{equation}
	In the first case where \cref{equation_construction_1} holds, we will use \cref{equation_given_info} to get the following:
	\begin{multline*}
		\biggl(\frac{2h_l+w_b}{\exp \varepsilon },
		\frac{w_l}{2\exp \varepsilon }+
			\max\Bigl(\frac{w_l}{2\exp \varepsilon }, \frac{h_b}{\exp \varepsilon }+\frac{1}{2}\Bigr)\biggr)
			\\ \in N(T').
	\end{multline*}
	We will see that we can pick sufficiently large $d$ so that we can guarantee
	$w \geq \frac{2h_l+w_b}{\exp \varepsilon }$ and
	$h \geq \frac{w_l}{2\exp \varepsilon }+\max(\frac{w_l}{2\exp \varepsilon }, \frac{h_b}{\exp \varepsilon }+\frac{1}{2})$ simultaneously.
	Assume otherwise.
	Since $\varepsilon >0$,
	$w \geq \frac{2h_l+w_b}{\exp \varepsilon }$ always,
	so
	$\frac{w_l}{2}+\max(\frac{w_l}{2}, h_b) < \frac{w_l}{2\exp \varepsilon }+\max(\frac{w_l}{2\exp \varepsilon }, \frac{h_b}{\exp \varepsilon }+\frac{1}{2})$,
	so
	$\max(\frac{w_l}{2}, h_b) < \max(\frac{w_l}{2\exp \varepsilon }, \frac{h_b}{\exp \varepsilon }+\frac{1}{2})$.
	Therefore $\frac{w_l}{2} < \frac{h_b}{\exp \varepsilon }+\frac{1}{2}$, so
	$\max(\frac{w_l}{2\exp \varepsilon }, \frac{h_b}{\exp \varepsilon }+\frac{1}{2}) = \frac{h_b}{\exp \varepsilon }+\frac{1}{2}$,
	so $h_b < \frac{h_b}{\exp \varepsilon }+\frac{1}{2} \Longleftrightarrow h_b < \frac{1}{2 (1-\exp(-\varepsilon )))}$.
	Since $(w_b, h_b) \in \Delta (T, d+\varepsilon )$ and because $C(\{ 1, 1 \}) \leq T$,
	by picking $d$ large enough, we can make $h_b < \frac{1}{2 (1-\exp(-\varepsilon )))}$ impossible.
	Note that $d$ still only depends on $\varepsilon $ and not on $h_b$.

	In the other case, we use \cref{equation_given_info} to get
	\[ \biggl(\max\Bigl(\frac{2 h_l}{\exp \varepsilon }+1, \frac{w_b}{\exp \varepsilon }\Bigr), \frac{w_l+h_b}{\exp \varepsilon }\biggr) \in N(T'). \]
	Proceed similarly.
\end{proof}

We now explain how the set $T$ was found.
Consider the set $S$ in \cref{sec_lower_bound}.
With values of $\sigma $, $\varepsilon $ and $\delta $ satisfying \cref{equation_for_exact_constants} exactly,
we get $S \leq \Delta (N^\infty (S), -\delta )$.
Define the operator $P(S) = \Delta (N^\infty (S), -\delta )$
and let $P^k (S)$ be $P$ applied $k$ times on $S$,
then we get $S \leq P(S)$, which means $P^k (S) \leq P^{k+1}(S)$ for all $k$,
and we conjecture the following based on numerical evidence:
\begin{conj}
	$\lim_{k \rightarrow \infty } P^k (S) = \{ (w, h) \mid h \geq f(w) \}$
	for an analytic function $f$ such that $f(\exp \omega ) - (\varepsilon -\omega \cdot \sigma ) \rightarrow 0$ as $\omega \rightarrow -\infty $
	and $f(\exp \omega ) - (\varepsilon -\omega /\sigma ) \rightarrow 0$ as $\omega \rightarrow +\infty $.
	Furthermore, let $P^\infty (S)$ denote that limiting set,
	then $P(P^\infty (S)) = P^\infty (S)$.
\end{conj}
Such a function $f$ must satisfy a certain functional equation; however, we are unable to solve it
or prove that the solution exists.
If the conjecture holds, pick $T = \Delta (P^\infty (S), d) \cap C(\{ w, h \})$ for some large $d$ and suitable $w$ and $h$,
we get $N^\infty (T) \leq \Delta (T, \delta )$, which would have implied the area is $\in \widetilde{\Theta }(n^{2\delta /\log 3})$.
However, we are only able to find a numerical approximation of $T$ that gives the bound $O(n^{1.032})$.
\iftrue
The algorithm computes $P^k (S)$ for a sufficiently large number $k$, with some approximations to make the
time complexity manageable.
The code to compute the set $T$ that certifies this upper bound can be found in the arXiv source code.
\fi

\section{Conclusion}

In this article, we improve the upper bound on the minimum area required
for a straight-line orthogonal drawing of the complete ternary tree,
and prove an almost-matching lower bound in the special case of
drawings satisfying the subtree separation property.

There are still several interesting open questions that need further research,
namely whether there is a nontrivial lower bound when the drawing is not required
to satisfy the subtree separation property, and whether the constant
can be determined analytically to prove that our lower bound is tight.

{
\hfuzz=10pt
\bibliographystyle{abbrv}
\bibliography{ternary-tree-on-grid}

\begin{thebibliography}{1}

\bibitem{Ali2015}
A.~Ali.
\newblock Straight line orthogonal drawings of complete ternary trees.
\newblock Technical report, MIT Summer Program in Undergraduate Research, July
  2015.

\bibitem{Chan_2019}
T.~M. Chan.
\newblock Tree drawings revisited.
\newblock {\em Discrete \& Computational Geometry}, 63(4):799–820, June 2019.

\bibitem{Covella_2021}
B.~Covella, F.~Frati, and M.~Patrignani.
\newblock On the area requirements of planar straight-line orthogonal drawings
  of ternary trees.
\newblock {\em Theoretical Computer Science}, 852:197–211, Jan. 2021.

\bibitem{Crescenzi_1992}
P.~Crescenzi, G.~Di~Battista, and A.~Piperno.
\newblock A note on optimal area algorithms for upward drawings of binary
  trees.
\newblock {\em Computational Geometry}, 2(4):187–200, Dec. 1992.

\bibitem{Frati}
F.~Frati.
\newblock {\em Straight-Line Orthogonal Drawings of Binary and Ternary Trees},
  page 76–87.
\newblock Springer Berlin Heidelberg, 2007.

\bibitem{Valiant_1981}
L.~G. Valiant.
\newblock Universality considerations in {VLSI} circuits.
\newblock {\em IEEE Transactions on Computers}, C–30(2):135–140, Feb. 1981.

\end{thebibliography}

}

\end{document}